\documentclass[submission,copyright,creativecommons,sharealike]{eptcs}
\pdfoutput=1
\providecommand{\event}{ACT 2022}

\usepackage[utf8]{inputenc}
\usepackage{amsmath}
\usepackage{amsthm}
\usepackage{amssymb}
\usepackage{physics}
\usepackage{tikz}
\usepackage{tikz-cd}
\usepackage{Cobordism}
\usepackage{tikzit}
\usepackage{circuitikz}
\usepackage[numbers]{natbib}
\usepackage{hyperref}
\usepackage[T1]{fontenc}
\usepackage{url}

\pgfsetlayers{bottom,background,main,internal,foreground,label,cobordism,top,selectionbox,edgelayer,nodelayer}

\input{styles.tikzdefs}
\tikzstyle{discarding}=[fill=white, draw=black, shape=circle, style=upground]
\tikzstyle{smalldiscarding}=[fill=white, draw=black, style=upground, scale=0.6]
\tikzstyle{backdiscard}=[fill=white, draw=black, shape=circle, style=downground, scale=0.5]
\tikzstyle{smallbackdiscard}=[fill=white, draw=black, shape=circle, style=downground, scale=0.5]
\tikzstyle{state}=[fill=white,rounded corners, draw=black, style=triang, tikzit shape=rectangle]
\tikzstyle{state_small}=[fill=white, draw=black, style={triang_small}, tikzit shape=rectangle]
\tikzstyle{state_lesspad}=[fill=white, draw=black, style={triang_lesssep}, tikzit shape=rectangle]
\tikzstyle{kstate}=[fill=white, draw=black, style=kpoint, tikzit shape=rectangle]
\tikzstyle{kstateconj}=[fill=white, draw=black, style=kpoint conjugate, tikzit shape=rectangle]
\tikzstyle{kstateBIG}=[fill=white, draw=black, style=big kpoint, tikzit shape=rectangle]
\tikzstyle{effect}=[fill=white,rounded corners, draw=black, style=triangdag]
\tikzstyle{effect_small}=[fill=white, draw=black, style={triangdag_small}]
\tikzstyle{keffect}=[fill=white, draw=black, style=kpoint adjoint]
\tikzstyle{keffectconj}=[fill=white, draw=black, style=kpoint transpose]
\tikzstyle{morphdag}=[style=mapdag]
\tikzstyle{morph}=[style=map]
\tikzstyle{morphtrans}=[style=maptrans]
\tikzstyle{morphconj}=[style=mapconj]
\tikzstyle{CPMmorph}=[style=dmap]
\tikzstyle{CPMmorphconj}=[style=dmapconj]
\tikzstyle{CPMmorphdag}=[style=dmapdag]
\tikzstyle{CPMmorphtrans}=[style=dmaptrans]
\tikzstyle{CPMstate}=[fill=white, draw=black, style=triang, doubled]
\tikzstyle{CPMstateBIG}=[fill=white, draw=black, style={triang_lesssep}, doubled]
\tikzstyle{CPMkstate}=[fill=white, draw=black, style=kpoint, tikzit shape=rectangle, doubled]
\tikzstyle{CPMkstateconj}=[fill=white, draw=black, style=kpoint conjugate, tikzit shape=rectangle, doubled]
\tikzstyle{CPMkstateBIG}=[fill=white, draw=black, style=big kpoint, tikzit shape=rectangle, doubled]
\tikzstyle{CPMkeffect}=[fill=white, draw=black, style=kpoint adjoint, doubled]
\tikzstyle{CPMkeffectconj}=[fill=white, draw=black, style=kpoint transpose, doubled]
\tikzstyle{UHfB}=[fill=white, draw=black, style=triangdag, doubled, inner sep=-2pt]
\tikzstyle{leak}=[style=tinypoint, regular polygon rotate=-90]
\tikzstyle{leakfill}=[style=tinypoint, regular polygon rotate=-90, fill=black]
\tikzstyle{Z}=[style=dot, fill=green]
\tikzstyle{X}=[style=dot, fill=red]
\tikzstyle{black_dot}=[style=dot, fill=black]
\tikzstyle{white_dot}=[style=dot, fill=white]
\tikzstyle{qblack_dot}=[style=ddot, fill=black]
\tikzstyle{qwhite_dot}=[style=ddot, fill=white]
\tikzstyle{whitephase}=[style=wphase dot, fill=white]
\tikzstyle{qredphase}=[style=phase dot, fill=red]
\tikzstyle{qgreenphase}=[style=phase dot, fill=green]
\tikzstyle{CPMbox}=[style=hadamard, doubled]
\tikzstyle{box}=[style=hadamard,rounded corners]
\tikzstyle{tallbox}=[style=hadamard, minimum height=7mm]
\tikzstyle{roundedtallbox}=[style=hadamard, rounded corners, minimum height=7mm]
\tikzstyle{bigbox}=[style=hadamard,rounded corners, minimum height=5mm, minimum width=7mm]
\tikzstyle{widebigbox}=[style=hadamard, minimum height=6mm, minimum width=18mm]
\tikzstyle{WIDEBOI}=[style=hadamard, minimum height=6mm, minimum width=40mm]
\tikzstyle{WIDEstate}=[fill=white, draw=black, style=triang, tikzit shape=rectangle, minimum width=40mm]
\tikzstyle{bigboxCPM}=[style=hadamard, minimum height=6mm, minimum width=8mm, doubled]
\tikzstyle{antipode}=[style=anti]
\tikzstyle{WIDEcutstate}=[style={state_doublecut}, text depth=2mm, minimum width=8mm]
\tikzstyle{WIDERcutstate}=[style={state_doublecut}, text depth=0mm, minimum width=36mm]
\tikzstyle{WIDEcutstateCPM}=[style={state_doublecut}, text depth=2mm, minimum width=8mm, doubled]
\tikzstyle{WIDEcuteffect}=[style={effect_doublecut}, text depth=2mm, minimum width=8mm]

\tikzstyle{dottededge}=[-, dotted]
\tikzstyle{double edge}=[-, style=doubled, draw=black, tikzit draw={rgb,255: red,191; green,0; blue,64}]

\usetikzlibrary{circuits.ee.IEC}
\usetikzlibrary{calc}

\newcommand{\morph}[1]{\xrightarrow{#1}}

\newcommand{\cat}[1]{\mathcal{#1}}
\newcommand{\cpm}{\mathsf{CPM}}
\newcommand{\fhilb}{\mathsf{FHilb}}

\newcommand{\optic}{\mathsf{Optic}}
\newcommand{\doptic}{\mathsf{Optic}^\dag}

\newcommand{\comb}{\mathsf{Comb}}
\newcommand{\dcomb}{\mathsf{Comb}^\dag}
\newcommand{\COMB}{\mathsf{COMB}}
\newcommand{\OPTIC}{\mathsf{OPTIC}}

\newcommand{\excl}{\mathord{!}}
\newcommand{\simcomb}{\sim_{\text{comb}}}
\newcommand{\simopt}{\sim_{\text{opt}}}
\newcommand{\simtau}{\sim_{\tau}}
\newcommand{\simsigma}{\sim_{\sigma}}

\newcommand{\tikzfigscale}[2]{\scalebox{#1}{\input{#2.tikz}}}

\makeatletter
\newcommand{\xRightarrow}[2][]{\ext@arrow 0359\Rightarrowfill@{#1}{#2}}
\makeatother

\makeatletter
\def\namedlabel#1#2{\begingroup
   \def\@currentlabel{#2}%
   \label{#1}\endgroup
}
\makeatother
\theoremstyle{definition}
\newtheorem{defn}{Definition}
\theoremstyle{plain}
\newtheorem{prop}{Proposition}
\theoremstyle{plain}
\newtheorem{lem}{Lemma}
\theoremstyle{plain}

\theoremstyle{plain}

\theoremstyle{remark}
\newtheorem*{remark}{Remark}
\theoremstyle{definition}
\newtheorem{example}{Example}
\theoremstyle{definition}
\newtheorem{counterexample}{Counterexample}

\newcommand{\Sym}{\mathsf{Sym}}
\newcommand{\Mon}{\mathsf{Mon}}
\newcommand{\Cat}{\mathsf{Cat}}

\title{Coend Optics for Quantum Combs}
\author{James Hefford
\institute{University of Oxford, UK}
\and
Cole Comfort
\institute{University of Oxford, UK}
}

\def\titlerunning{Coend Optics for Quantum Combs}
\def\authorrunning{J. Hefford \& C. Comfort}

\begin{document}

\maketitle

\begin{abstract}
  We compare two possible ways of defining a category of 1-combs, the first intensionally as coend optics and the second extensionally as a quotient by the operational behaviour of 1-combs on lower-order maps.
  We show that there is a full and bijective on objects functor quotienting the intensional definition to the extensional one and give some sufficient conditions for this functor to be an isomorphism of categories.
  We also show how the constructions for 1-combs can be extended to produce polycategories of $n$-combs with similar results about when these polycategories are equivalent.
  The extensional definition is of particular interest in the study of quantum combs and we hope this work might produce further interest in the usage of optics for modelling these structures in quantum theory.
\end{abstract}

\section{Introduction}
\label{sec:intro}
The traditional way in which physical systems are modelled is by considering a state space which evolves according to processes which act on that space.
For example, a quantum circuit is traditionally viewed in terms of linear operators being applied to a Hilbert space; electrical circuits in terms of certain operators acting on phase space; probabilistic theories in terms of stochastic maps acting on probability spaces.

This approach has proven to be amenable to categorical analysis.
For example, the ZX-calculus \cite{coecke_zx,zxwcs}, graphical affine algebra \cite{gaa,2106.07763,glagr} and markov categories \cite{fritz} have all been successful in formally modelling these respective classes of systems using the theory of monoidal categories.
Moreover, categorical quantum mechanics \cite{abramsky_cqm,coecke_picturalism,gogioso_cpt} and the framework of generalised/operational probabilistic theories \cite{barrett_gpts,chiribella_probabilistic} provide semantics for modelling more general quantum-like theories.

However, the approach of modelling systems merely in terms of the action of operators on the state space may not fully capture the behaviour of the system.
When the collection of operators is itself regarded as the state space, this traditional approach gives little insight into the evolution of this new, ``higher order'' state space.
What is missing is a theory of second order processes, a theory of processes which themselves act on (first order) processes.
Or indeed a theory of $n^\text{th}$ order processes which act on $(n-1)^\text{th}$ order processes.

In the theory of quantum circuits, the domain which we are chiefly interested in this paper, these higher order processes are known as quantum supermaps \cite{chiribella_supermaps,kissinger_caus}.
In their most abstract formulation, it is known that there exist quantum supermaps, such as the quantum \texttt{switch}, which go beyond the standard quantum circuit model by not possessing a factorisation as a circuit with definite causal ordering of gates and no time-loops \cite{chiribella_switch,chiribella_switch2}.
Yet there is a class of supermaps which can be adequately modelled by ``circuits with holes'' \cite{chiribella_circuits, chiribella_networks} where one has a framework quantum circuit with slots that can be filled with first-order maps.
Indeed, all second-order deterministic single-party supermaps on quantum channels possess a factorisation as a circuit \cite{chiribella_supermaps,eggeling,kissinger_caus}.

In this paper we restrict our attention to these ``circuits with holes'', otherwise known as $n$-combs \cite{chiribella_circuits, chiribella_networks}.
For example, 1-combs are often drawn suggestively as diagrams of the form:

\begin{equation}\label{eq:comb}
  \tikzfigscale{0.9}{figs/comb_fact}
\end{equation}

Some care is needed though to make these drawings rigorous and to demonstrate that a suitable (possibly symmetric monoidal) category of combs can be defined.
In much of the quantum literature it is assumed that the base category of first order processes is compact closed, or at least embeds into one.
In this case it is possible to bend input and output wires to express combs as maps without holes and use the drawing \eqref{eq:comb} in an unambiguous way; for example, see \cite{chiribella_circuits,kissinger_caus}.
Outside of the quantum literature there are approaches to defining comb diagrams without the assumption of closure \cite{roman_coend,roman_comb}, but it is not clear when this coincides with the quantum definition.

In this article we focus on defining categories of combs without any assumptions of closure on the category of first order processes.
In Section \ref{sec:combs}, we compare two constructions which take an arbitrary symmetric monoidal category and produce a symmetric monoidal category of combs.
Both of these constructions represent a comb as a pair of morphisms $(f,g)$ from the theory of first order processes, quotiented by their behaviour on first-order processes.

The first construction, which we define in Subsection \ref{subsec:extensionalcombs}, $\comb:\Sym\Mon\Cat\to \Sym\Mon\Cat$, quotients combs by their extensional behaviour: two combs are equal when they produce the same output on all first-order inputs.
In other words this identifies two combs when they appear to be the same when probed with all first order processes $\lambda$:
\begin{equation*}
  (f,g)\simcomb (f',g') \hspace{1cm} \text{when} \hspace{1cm} \tikzfigscale{0.9}{figs/intro_comb} \hspace{1cm} \forall \lambda
\end{equation*}
This equivalence relation has been discussed before \cite{coecke_resources} \footnote{we note that our category of combs is distinct from that developed there: the objects of their category being different than those studied in this document} and is perhaps the one that would be most immediate to those studying quantum theory.

The second construction, which we review in Subsection \ref{subsec:optics},  is that of the category of coend optics, $\optic:\Sym\Mon\Cat\to\Sym\Mon\Cat$ (which we shall henceforth just call optics) \cite{clarke_profunctor,riley_optics,roman_optics,pastro_street}.
Optics are used to encompass  bidirectional data accessors familiar to the computer science community such as lenses, prisms and grates, amongst many others.
Their usage to model combs and more general ``circuits with holes'' has been described in \cite{roman_coend,roman_comb}.
In contrast to the previous construction this quotients the combs by their intensional behaviour, allowing first-order maps to slide along the shared environment connecting the two factors together:
\begin{equation*}
  \begin{tikzpicture}[baseline={([yshift=-.5ex]current bounding box.center)}]
    \node[Pants, top] (pants) {};
    \node[Cyl, anchor=top] (tube) at (pants.leftleg){};
    \node[Copants, bot, anchor=leftleg] (copants) at (tube.bot) {};
    \node[Top3D] at (copants.rightleg) {};
    \node[Bot3D] at (pants.rightleg) {};
    \node[Bot3D] at (tube.center) {};
    \begin{scope}[internal string scope]
      \node[sq tiny label] (g) at (pants.center) {$g$};
      \node[sq tiny label] (f) at (copants.center) {$f$};
      \node[sq tiny label] (v) at (copants.leftleg) {$v$};
      \draw (f.center) to (copants.belt);
      \draw (g.center) to (pants.belt);
      \draw (f.center) to [out=165,in=-90] (copants.leftleg) to (tube.top) to [out=90,in=-165] (g.center);
      \draw (f.center) to [out=15,in=-90] (copants.rightleg);
      \draw (g.center) to [out=-15,in=90] (pants.rightleg);
    \end{scope}
  \end{tikzpicture}
  \quad
  \simopt
  \quad
  \begin{tikzpicture}[baseline={([yshift=-.5ex]current bounding box.center)}]
    \node[Pants, top] (pants) {};
    \node[Cyl, anchor=top] (tube) at (pants.leftleg){};
    \node[Copants, bot, anchor=leftleg] (copants) at (tube.bot) {};
    \node[Top3D] at (copants.rightleg) {};
    \node[Bot3D] at (pants.rightleg) {};
    \node[Bot3D] at (tube.center) {};
    \begin{scope}[internal string scope]
      \node[sq tiny label] (g) at (pants.center) {$g$};
      \node[sq tiny label] (f) at (copants.center) {$f$};
      \node[sq tiny label] (v) at (pants.leftleg) {$v$};
      \draw (f.center) to (copants.belt);
      \draw (g.center) to (pants.belt);
      \draw (f.center) to [out=165,in=-90] (copants.leftleg) to (tube.top) to [out=90,in=-165] (g.center);
      \draw (f.center) to [out=15,in=-90] (copants.rightleg);
      \draw (g.center) to [out=-15,in=90] (pants.rightleg);
    \end{scope}
  \end{tikzpicture}
\end{equation*}

In Subsection \ref{subsec:equivalence} we show that there is always a full and bijective on objects monoidal functor from optics to the extensional definition, $\optic(\cat{C})\morph{}\comb(\cat{C})$.
We then give some sufficient conditions for this functor to exhibit an isomorphism of symmetric monoidal categories.
In particular we show that when the category of first-order processes is cartesian and there exists a state for every type or when it is compact closed, the two definitions coincide.
We also show that in the case of the category of unitaries between Hilbert spaces, the definitions again coincide.
This case (alongside compact closed categories) is particularly important for quantum theory.
We leave it as future work to fully characterise when $\optic(\cat{C})\cong\comb(\cat{C})$ and note that there are important cases of combs not covered by the sufficient conditions proven in this work.

In section \ref{sec:CPM} we specialise to the case where the base category $\cat C$ is $\dag$-compact closed  and restrict to the subcategory where the  maps constituting the combs are the daggers of each other,  that is where $g=f^\dag$ in \eqref{eq:comb}.
This subcategory of the extensional category of combs collapses to the ${\sf CPM}$-construction \cite{selinger_cpm} which is used to generate a category of generalised completely positive maps from some underlying category.

In the final section \ref{sec:ncombs} we turn our attention to $n$-combs, which map $n$ first-order processes into first-order process.
These combs naturally form a polycategory and we show that in the presence of compact closure, the extensional and intensional definitions once again coincide.

\section{Combs}
\label{sec:combs}

In this section we define two notions of factorisable 1st order single input combs.
These categories are given by functors $\Sym\Mon\Cat \to\Sym\Mon\Cat$ whose morphisms are given by pairs of maps composable along an interface, as per the right hand side of \eqref{eq:comb}.
In Subsection \ref{subsec:extensionalcombs} we establish the extensional definition; in Subsection \ref{subsec:optics} we review the intensional definition of optics; and in Subsection \ref{subsec:equivalence} we give sufficient conditions under which both definitions coincide.

\subsection{Extensional combs}\label{subsec:extensionalcombs}

Let us begin by considering possible extensional definitions of combs.
Firstly, one could ask that the combs are equal as morphisms in the original category when we extend their inputs:
\begin{equation}\label{eq:braidequiv}
  (f,g)\simsigma(f',g') \quad \iff \quad \tikzfigscale{0.9}{figs/comb_braid}
\end{equation}
While this is an equivalence relation on pairs of morphisms, it is not a congruence with respect to composition.
Suppose $(f,g)\simsigma(f',g')$ and $(h,k)\simsigma(h',k')$.
Then $(h,k)\circ(f,g) = ((1\otimes h)f,g(1\otimes k))\simsigma ((1\otimes h')f,g(1\otimes k')) = (h',k')\circ(f,g)$ which is not in general equivalent to $(h',k')\circ(f',g')$.

We could instead ask that two combs are equivalent if they are equal on all inputs to the comb:

\begin{equation*}
  (f,g)\simtau (f',g') \quad \iff \quad \tikzfigscale{0.9}{figs/comb_equivalence_bad}
\end{equation*}

This also forms an equivalence relation on pairs of morphisms, although it is too coarse.
Consider the free symmetric monoidal category generated by one object $A$, two states $\phi,\psi:I\to A$ and an effect $\excl:A\to I$ such that $\excl\circ\phi = \excl\circ\psi=1_I$.
Then $(1_I \otimes \psi, 1_I \otimes \excl) \simtau (1_I \otimes \phi, 1_I \otimes \excl)$; however evaluating these combs on the braid one finds,

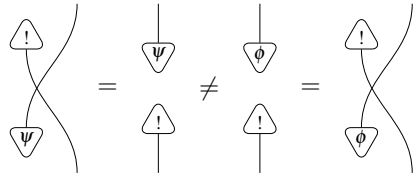
\begin{equation*}
  \tikzfigscale{0.9}{figs/comb_equiv_fail}
\end{equation*}

So if we want comb to behave compatibly with the monoidal structure of the category, we need something stronger than equality on all inputs.

\begin{defn}[Extensional Comb Equivalence]
We say that two combs are equivalent if they are equal on all extended inputs:

\begin{equation}\label{eq:combequiv}
 (f,g)_E \simcomb (f',g')_{E'}\ \iff\  \tikzfigscale{0.9}{figs/comb_equivalence}
\end{equation}
\end{defn}

This definition subsumes both of the previous definitions, but in the compact closed case \eqref{eq:braidequiv} is sufficient to recover the full extensional equivalence.

\begin{prop}
  When $\cat{C}$ is compact closed $(f,g)\simcomb (f',g') \iff (f,g)\simsigma (f',g')$.
\end{prop}
\begin{proof}
  The forwards direction is immediate.
  The backwards direction follows by graphical manipulation:
  \begin{equation*}
    \tikzfigscale{0.9}{figs/braid_comb}
  \end{equation*}
\end{proof}

\begin{defn}
  Given a symmetric monoidal category $\cat C$, the symmetric monoidal category of extensional combs $\comb(\cat{C})$ has:

\begin{description}
\item[Objects:] pairs $(A,A')$ of objects of $\cat{C}$.
\item[Morphisms:]
   $(f,g):(A,A')\morph{}(B,B')$ are equivalence classes of pairs of morphisms $f:A\morph{}E\otimes B$ and $g:E\otimes B'\morph{}A'$ of $\cat{C}$ under the comb equivalence relation $\simcomb$.

  Composition of morphisms is given by $(f',g')\circ(f,g) = ( (1\otimes f')f, g(1\otimes g'))$.

\item[Monoidal structure:]
  On objects $(A,A')\otimes (B,B') = (A\otimes B,A'\otimes B')$ and on morphisms:
  \begin{equation*}
    \tikzfigscale{0.9}{figs/comb_tensor}
  \end{equation*}

  The unit object is $(I,I)$ with structural isomorphisms given by $(\lambda,\lambda^{-1}):(A,A')\otimes (I,I) = (A\otimes I,A'\otimes I) \morph{} (A,A')$ and $(\rho,\rho^{-1})$.

  The symmetry is defined similarly.
\end{description}
\end{defn}

\begin{lem}
$\comb$ defines a functor $\Sym\Mon\Cat\to\Sym\Mon\Cat$.
\end{lem}

\subsection{Optics}
\label{subsec:optics}

Optics provide another potential definition of combs; albeit an intensional one, as opposed to the extensional one described in the previous subsection.

We will use the graphical calculus of internal string diagrams/pointed profunctors to work with optics.
Internal string diagrams were first introduced in the $\mathsf{Vect}$-enriched case in \cite{vicary_bordism} and further explored in \cite{hu_tubes}.
The same sort of  graphical calculus was described in \cite{roman_coend} where the author shows that they form a 2-category of pointed profunctors.

Internal string diagrams consist of usual string diagrams for monoidal categories bounded inside cobordisms.
For example the identity, contravariant and covariant embeddings of the tensor product and tensor unit are drawn as follows:

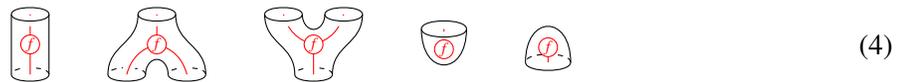
\begin{equation}
\label{eq:cobgens}
  \begin{tikzpicture}[baseline={([yshift=-.5ex]current bounding box.center)}]
    \node[Cyl,top,bot] (tube) {};
   \begin{scope}[internal string scope]
     \node[sq tiny label] (f) at (tube.center) {$f$};
     \draw (tube.bottom) to (f.center) to (tube.top);
   \end{scope}
  \end{tikzpicture}
  \qquad
  \begin{tikzpicture}[baseline={([yshift=-.5ex]current bounding box.center)}]
    \node[Pants,top,bot] (pants) {};
   \begin{scope}[internal string scope]
    \node[sq tiny label] (f) at (pants.center) {$f$};
    \draw (pants.belt) to (f.center);
    \draw[bend left] (pants.leftleg) to (f.center);
    \draw[bend right] (pants.rightleg) to (f.center);
   \end{scope}
  \end{tikzpicture}
  \qquad
  \begin{tikzpicture}[baseline={([yshift=-.5ex]current bounding box.center)}]
    \node[Copants,top,bot] (copants) {};
   \begin{scope}[internal string scope]
    \node[sq tiny label] (f) at (copants.center) {$f$};
    \draw (copants.belt) to (f.center);
    \draw[bend right] (copants.leftleg) to (f.center);
    \draw[bend left] (copants.rightleg) to (f.center);
   \end{scope}
  \end{tikzpicture}
\qquad
  \begin{tikzpicture}[baseline={([yshift=-.5ex]current bounding box.center)}]
    \node[Cup, top, scale=1.2]  (cup) {};
   \begin{scope}[internal string scope]
    \node[sq tiny label] (f) at ($(cup.center)+(0,-.27)$)  {$f$};
     \draw (f.center) to (cup.center);
   \end{scope}
  \end{tikzpicture}
\qquad
  \begin{tikzpicture}[baseline={([yshift=-.5ex]current bounding box.center)}]
    \node[Cap, bot, scale=1.2] (cap) {};
   \begin{scope}[internal string scope]
    \node[sq tiny label] (f) at ($(cap.center)+(0,.2)$) {$f$};
     \draw (f.center) to (cap.center);
   \end{scope}
  \end{tikzpicture}
\end{equation}

The internal diagrams can be manipulated and composed as usual, but they are constrained by the topology of the cobordisms.
Moreover, when we compose these diagrams together, we are allowed to slide morphisms between them as follows:

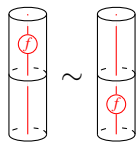
\begin{equation*}
  \begin{tikzpicture}[baseline={([yshift=-.5ex]current bounding box.center)}]
    \node[Cyl,top,bot] (tube) {};
    \node[Cyl,bot,anchor=top] (tube2) at (tube.bot) {};
    \begin{scope}[internal string scope]
      \node[sq tiny label] (f) at (tube.center) {$f$};
      \draw (tube2.bot) to (f) to (tube.top);
    \end{scope}
  \end{tikzpicture}
  \ \sim \
  \begin{tikzpicture}[baseline={([yshift=-.5ex]current bounding box.center)}]
    \node[Cyl,top,bot] (tube) {};
    \node[Cyl,bot,anchor=top] (tube2) at (tube.bot) {};
    \begin{scope}[internal string scope]
      \node[sq tiny label] (f) at (tube2.center) {$f$};
      \draw (tube2.bot) to (f) to (tube.top);
    \end{scope}
  \end{tikzpicture}
\end{equation*}

The shapes in \eqref{eq:cobgens} are associative monoids and comonoids and there exist the following 2-cells which allow us to ``pop bubbles'' (as well as some extra coherence conditions):

\begin{equation}\label{eq:tube2cells}
  \begin{tikzpicture}[baseline={([yshift=-.5ex]current bounding box.center)}]
    \node[Pants, top] (pants) {};
    \node[Copants, bot, lowercob, anchor=leftleg] (copants) at (pants.leftleg) {};
   % \node[Top3D] at (copants.rightleg) {};
    \node[Bot3D] at (pants.rightleg) {};
    \node[Bot3D] at (pants.leftleg) {};
   \begin{scope}[internal string scope]
     \node[sq tiny label] (f) at (pants.center) {$f$};
     \node[sq tiny label] (g) at (copants.center) {$g$};
     \draw (f.center) to (pants.belt);
     \draw[bend right] (f.center) to (pants.leftleg);
     \draw[bend left] (f.center) to (pants.rightleg);
     \draw (g.center) to (copants.belt);
     \draw[bend left] (g.center) to (copants.leftleg);
     \draw[bend right] (g.center) to (copants.rightleg);
   \end{scope}
  \end{tikzpicture}
\Rightarrow
  \begin{tikzpicture}[baseline={([yshift=-.5ex]current bounding box.center)}]
    \node[Cyl,xscale=1.2,top,anchor=bot] (tube) {};
    \node[Cyl,xscale=1.2,bot,anchor=top] (tube1) at (tube.bot) {};
    \begin{scope}[internal string scope]
     \node[sq tiny label] (f) at (tube.center) {$f$};
     \node[sq tiny label] (g) at (tube1.center) {$g$};
     \draw (tube1.bot) to (g.center);
     \draw (f.center) to (tube.top);
     \draw[bend left] (f.center) to (g.center);
     \draw[bend right] (f.center) to (g.center);
    \end{scope}
  \end{tikzpicture}\ ,
\hspace*{.5cm}
  \begin{tikzpicture}[baseline={([yshift=-.5ex]current bounding box.center)}]
    \node[Cyl,top,anchor=bot] (tube) {};
    \node[Cyl,bot,anchor=top] (tube1) at (tube.bot) {};
    \begin{scope}[internal string scope]
     \node[sq tiny label] (f) at (tube.bot) {$f$};
     \draw (f.center) to (tube1.bot);
     \draw (f.center) to (tube.top);
    \end{scope}
  \end{tikzpicture}\
  \begin{tikzpicture}[baseline={([yshift=-.5ex]current bounding box.center)}]
    \node[Cyl,top,anchor=bot] (tube) {};
    \node[Cyl,bot,anchor=top] (tube1) at (tube.bot) {};
    \begin{scope}[internal string scope]
     \node[sq tiny label] (g) at (tube.bot) {$g$};
     \draw (g.center) to (tube1.bot);
     \draw (g.center) to (tube.top);
    \end{scope}
  \end{tikzpicture}
\Rightarrow
  \begin{tikzpicture}[baseline={([yshift=-.5ex]current bounding box.center)}]
    \node[Pants,xscale=1.2,bot] (pants) {};
    \node[Copants,xscale=1.2, top, anchor=belt] (copants) at (pants.belt) {};
   % \node[Top3D] at (copants.rightleg) {};
   % \node[Bot3D] at (pants.rightleg) {};
    \node[Bot3D,xscale=1.2] at (pants.belt) {};
    \begin{scope}[internal string scope]
    \node[sq tiny label] (f) at ($(pants.belt)+(-0.15,.15)$) {$f$};
     \node[sq tiny label] (g) at ($(pants.belt)+(0.15,.15)$) {$g$};
     \draw[in=-90, out=90, looseness=1.3]  (f.center) to ($(copants.leftleg)+(0,0)$);
     \draw[in=-90, out=90, looseness=1.3]  (g.center) to ($(copants.rightleg)+(0,0)$);
     \draw[in=90, out=-90, looseness=1.3]  (f.center) to ($(pants.leftleg)+(0,0)$);
     \draw[in=90, out=-90, looseness=1.3]  (g.center) to ($(pants.rightleg)+(0,0)$);
    \end{scope}
  \end{tikzpicture}
\ ,
\hspace*{.5cm}
  \begin{tikzpicture}[baseline={([yshift=-.5ex]current bounding box.center)}]
    \node[Cup, top,scale=1.2] (cup) at (0,1.5) {};
    \node[Cap, bot,scale=1.2] (cap) at (0,0) {};
   % \node[Top3D] at (copants.rightleg) {};
    \begin{scope}[internal string scope]
    \node[sq tiny label] (f) at ($(cap.center)+(0,.2)$) {$g$};
    \node[sq tiny label] (g) at ($(cup.center)+(0,-.27)$)  {$f$};
     \draw (g.center) to (cup.center);
     \draw (f.center) to (cap.center);
    \end{scope}
  \end{tikzpicture}
\Rightarrow
  \begin{tikzpicture}[baseline={([yshift=-.5ex]current bounding box.center)}]
    \node[Cyl,xscale=1.2,top,anchor=bot] (tube) {};
    \node[Cyl,xscale=1.2,bot,anchor=top] (tube1) at (tube.bot) {};
    \begin{scope}[internal string scope]
     \node[sq tiny label] (f) at (tube.center) {$f$};
     \node[sq tiny label] (g) at (tube1.center) {$g$};
     \draw (tube1.bot) to (g.center);
     \draw (f.center) to (tube.top);
    \end{scope}
  \end{tikzpicture}\ ,
\hspace*{.5cm}
  \begin{tikzpicture}[baseline={([yshift=-.5ex]current bounding box.center)}]
	\begin{pgfonlayer}{nodelayer}
		\node [style=none] (0) at (-1, 1) {};
		\node [style=none] (1) at (-1, 0) {};
		\node [style=none] (2) at (0, 0) {};
		\node [style=none] (3) at (0, 1) {};
	\end{pgfonlayer}
	\begin{pgfonlayer}{edgelayer}
		\draw[style=dashed] (2.center) to (3.center);
		\draw[style=dashed] (3.center) to (0.center);
		\draw[style=dashed] (0.center) to (1.center);
		\draw[style=dashed] (1.center) to (2.center);
	\end{pgfonlayer}
\end{tikzpicture}
\Rightarrow
  \begin{tikzpicture}[baseline={([yshift=-.5ex]current bounding box.center)}]
    \node[Cup] (cup) {};
    \node[Cap, bot]  at (cup) {};
   % \node[Top3D] at (copants.rightleg) {};
    \begin{scope}[internal string scope]
    \end{scope}
  \end{tikzpicture}
\end{equation}

There is much more to say about pointed profunctors, but we will omit the technical discussion and refer the interested reader to \cite{vicary_bordism}  and \cite{roman_coend} for a more in-depth discussion.

We are now in a position to give the definition of the category of optics.

\begin{defn}[Category of optics \cite{pastro_street,clarke_profunctor}]
  Given a symmetric monoidal category $\cat C$, the category of optics $\optic(\cat{C})$, has the same objects as $\comb(\cat{C})$.
  Morphisms are pairs $(f,g)_E$ like in $\comb(\cat{C})$ however, instead of quotienting the morphisms by the equivalence relation $\simcomb$, we quotient morphisms by the equivalence relation $\simopt$ imposed by embedding the combs inside the cobordisms:
  \begin{equation}\label{eq:optequiv}
    \begin{tikzpicture}[baseline={([yshift=-.5ex]current bounding box.center)}]
      \node[Pants, top] (pants) {};
      \node[Cyl, anchor=top] (tube) at (pants.leftleg){};
      \node[Copants, bot, anchor=leftleg] (copants) at (tube.bot) {};
      \node[Top3D] at (copants.rightleg) {};
      \node[Bot3D] at (pants.rightleg) {};
      \node[Bot3D] at (tube.center) {};
      \begin{scope}[internal string scope]
        \node[sq tiny label] (g) at (pants.center) {$g$};
        \node[sq tiny label] (f) at (copants.center) {$f$};
        \node[sq tiny label] (v) at (copants.leftleg) {$v$};
        \draw (f.center) to (copants.belt);
        \draw (g.center) to (pants.belt);
        \draw (f.center) to [out=165,in=-90] (copants.leftleg) to (tube.top) to [out=90,in=-165] (g.center);
        \draw (f.center) to [out=15,in=-90] (copants.rightleg);
        \draw (g.center) to [out=-15,in=90] (pants.rightleg);
      \end{scope}
    \end{tikzpicture}
    \quad
    \simopt
    \quad
    \begin{tikzpicture}[baseline={([yshift=-.5ex]current bounding box.center)}]
      \node[Pants, top] (pants) {};
      \node[Cyl, anchor=top] (tube) at (pants.leftleg){};
      \node[Copants, bot, anchor=leftleg] (copants) at (tube.bot) {};
      \node[Top3D] at (copants.rightleg) {};
      \node[Bot3D] at (pants.rightleg) {};
      \node[Bot3D] at (tube.center) {};
      \begin{scope}[internal string scope]
        \node[sq tiny label] (g) at (pants.center) {$g$};
        \node[sq tiny label] (f) at (copants.center) {$f$};
        \node[sq tiny label] (v) at (pants.leftleg) {$v$};
        \draw (f.center) to (copants.belt);
        \draw (g.center) to (pants.belt);
        \draw (f.center) to [out=165,in=-90] (copants.leftleg) to (tube.top) to [out=90,in=-165] (g.center);
        \draw (f.center) to [out=15,in=-90] (copants.rightleg);
        \draw (g.center) to [out=-15,in=90] (pants.rightleg);
      \end{scope}
    \end{tikzpicture}
  \end{equation}
  The string diagrams can be freely moved around the interior of the cobordism, but can not pass through the surface: as a result we are able to slide maps on the environment wire between the two halves with the equivalence relation generated by $( (v\otimes 1)f,g )_{E'} \sim (f,g(v\otimes 1))_E$.

  Composition, identities, and symmetric monoidal structure is as in $\comb(\cat{C})$.
  That $\simopt$ is a congruence and that the composite of two optics is another optic (i.e. that the composite of the comb-shaped cobordisms in \eqref{eq:optequiv} can be manipulated to give another comb-shaped cobordism) follows by a composition of the 2-cells in \eqref{eq:tube2cells}, see e.g. \cite{riley_optics} for more details.
\end{defn}

\subsection{Equivalence of the Definitions}\label{subsec:equivalence}
In this section we consider the question of when $\optic(\cat{C})$ and $\comb(\cat{C})$ are equivalent.
It is fairly straightforward to show that there is always a functor $\optic(\cat{C}) \morph{} \comb(\cat{C})$ turning the intensional combs into extensional combs.

\begin{prop}\label{prop:optic_comb}
  Given a symmetric monoidal category $\cat{C}$, there is a bijective on objects, full symmetric monoidal functor $\optic(\cat{C}) \morph{} \comb(\cat{C})$.
\end{prop}

\begin{proof}
For each $\lambda$ there is a mapping:
  \begin{equation*}
   \begin{tikzpicture}[baseline={([yshift=-.5ex]current bounding box.center)}]
     \node[Pants, top] (pants) {};
     \node[Cyl, anchor=top] (tube) at (pants.leftleg){};
     \node[Copants, bot, anchor=leftleg] (copants) at (tube.bot) {};
     \node[Top3D] at (copants.rightleg) {};
     \node[Bot3D] at (pants.rightleg) {};
     \node[Bot3D] at (tube.center) {};
     \begin{scope}[internal string scope]
       \node[sq tiny label] (g) at (pants.center) {$g$};
       \node[sq tiny label] (f) at (copants.center) {$f$};
       \node[sq tiny label] (v) at (copants.leftleg) {$v$};
       \draw (f.center) to (copants.belt);
       \draw (g.center) to (pants.belt);
       \draw (f.center) to [out=165,in=-90] (copants.leftleg) to (tube.top) to [out=90,in=-165] (g.center);
       \draw (f.center) to [out=15,in=-90] (copants.rightleg);
       \draw (g.center) to [out=-15,in=90] (pants.rightleg);
     \end{scope}
   \end{tikzpicture}
   \ \mapsto \
   \tikzfigscale{0.9}{figs/bubble_pop}
  \end{equation*}
  This preserves the sliding of morphisms $v$ along the ancillary wire.
\end{proof}
\begin{remark}
  Formally, the mapping above gives a cowedge for $\cat{C}(A,\mathord{-}\otimes B)\times \cat{C}(\mathord{=} \otimes B', A')$ and must therefore factor uniquely via the coend.
\end{remark}

It is not immediately obvious whether the functor of the previous proposition is faithful and thus witnesses an equivalence of categories.
\begin{counterexample}
  Consider the free commutative monoidal category generated by one object $A$ and a single idempotent $f:A\morph{}A$.
  Then $(1_A,f)_I \nsim_{\text{opt}} (f,1_A)_I$ but $(1_A,f)_I\simcomb (f,1_A)_I$ and thus $\optic(\cat{C})\ncong\comb(\cat{C})$ in this case.
\end{counterexample}

We now explore some classes of categories where there is an equivalence $\optic(\cat{C})\cong\comb(\cat{C})$.

\begin{prop}\label{prop:compactclosed}
  Given a compact closed category $\cat{C}$, there is a symmetric monoidal isomorphism of categories $\optic(\cat{C}) \cong \comb(\cat{C})$.
\end{prop}
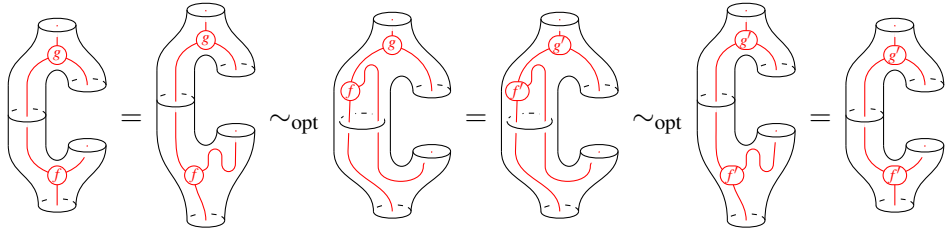
\begin{proof}
  \begin{equation*}
    \begin{tikzpicture}[baseline={([yshift=-.5ex]current bounding box.center)}]
      \node[Pants, top] (pants) {};
      \node[Cyl, anchor=top] (tube) at (pants.leftleg){};
      \node[Copants, bot, anchor=leftleg] (copants) at (tube.bot) {};
      \node[Top3D] at (copants.rightleg) {};
      \node[Bot3D] at (pants.rightleg) {};
      \node[Bot3D] at (tube.center) {};
      \begin{scope}[internal string scope]
        \node[sq tiny label] (g) at (pants.center) {$g$};
        \node[sq tiny label] (f) at (copants.center) {$f$};
        \draw (f.center) to (copants.belt);
        \draw (g.center) to (pants.belt);
        \draw (f.center) to [out=165,in=-90] (copants.leftleg) to (tube.top) to [out=90,in=-165] (g.center);
        \draw (f.center) to [out=15,in=-90] (copants.rightleg);
        \draw (g.center) to [out=-15,in=90] (pants.rightleg);
      \end{scope}
    \end{tikzpicture}
    \ = \
    \begin{tikzpicture}[baseline={([yshift=-.5ex]current bounding box.center)}]
      \node[Pants, top] (pants) {};
      \node[Cyl, anchor=top] (tube) at (pants.leftleg){};
      \node[Copants, bot, height scale = 1.5, anchor=leftleg] (copants) at (tube.bot) {};
      \node[Top3D] at (copants.rightleg) {};
      \node[Bot3D] at (pants.rightleg) {};
      \node[Bot3D] at (tube.center) {};
      \begin{scope}[internal string scope]
        \node[sq tiny label] (g) at (pants.center) {$g$};
        \node[sq tiny label] (f) at (copants.center) [left=.05\cobwidth] {$f$};
        \draw (f.center) to [in=90,out=-90] (copants.belt);
        \draw (g.center) to (pants.belt);
        \draw (f.center) to [out=165,in=-90] (copants.leftleg) to (tube.top) to [out=90,in=-165] (g.center);
        \draw (f.center) to [out=45,in=-90] ++(0.19,0.15) to ++(0,0.1) to [out=90,in=180] ++(0.09,0.1) to [out=0,in=90] ++(0.09,-0.1) to ++(0,-0.1) to [out=-90,in=180] ++(0.09,-0.1) to [out=0,in=-90] ++(0.09,0.1) to (copants.rightleg);
        \draw (g.center) to [out=-15,in=90] (pants.rightleg);
      \end{scope}
    \end{tikzpicture}
    \ \simopt \
    \begin{tikzpicture}[baseline={([yshift=-.5ex]current bounding box.center)}]
      \node[Pants, top, left leg scale = 1.5,height scale =1.2] (pants) {};
      \node[Cyl, top scale=1.5, bottom scale=1.5, anchor=top] (tube) at (pants.leftleg){};
      \node[Copants, left leg scale=1.5, anchor=leftleg] (copants) at (tube.bot) {};
      \node[Top3D] at (copants.rightleg) {};
      \node[Bot3D] at (pants.rightleg) {};
      \node[Bot3D] at (copants.belt) {};
      \node[Bot3D,scale=1.2] at (tube.center) {};
      \begin{scope}[internal string scope]
        \node[sq tiny label] (g) at (pants.center) [above] {$g$};
        \node[sq tiny label] (f) at (pants.leftleg) [left=.05\cobwidth] {$f$};
        \node (a) at (copants.leftleg) [left=.05\cobwidth] {};
        \draw (f.center) to (a.center) to [in=90,out=-90] (copants.belt);
        \draw (g.center) to (pants.belt);
        \draw (f.center) to [out=90,in=-165] (g.center);
        \draw (f.center) to [out=45,in=-90] ++(0.19,0.15) to ++(0,0.1) to [out=90,in=180] ++(0.09,0.1) to [out=0,in=90] ++(0.09,-0.1) to ++(0,-1.25) to [out=-90,in=180] ++(0.3,-0.2) to [out=0,in=-90] ++(0.3,0.2) to [out=90,in=-90] (copants.rightleg);
        \draw (g.center) to [out=-15,in=90] (pants.rightleg);
      \end{scope}
    \end{tikzpicture}
    \ = \
    \begin{tikzpicture}[baseline={([yshift=-.5ex]current bounding box.center)}]
      \node[Pants, top, left leg scale = 1.5,height scale =1.2] (pants) {};
      \node[Cyl, top scale=1.5, bottom scale=1.5, anchor=top] (tube) at (pants.leftleg){};
      \node[Copants, bot, left leg scale=1.5, anchor=leftleg] (copants) at (tube.bot) {};
      \node[Top3D] at (copants.rightleg) {};
      \node[Bot3D] at (pants.rightleg) {};
      \node[Bot3D,scale=1.2] at (tube.center) {};
      \begin{scope}[internal string scope]
        \node[sq tiny label] (g) at (pants.center) [above] {$g'$};
        \node[sq tiny label] (f) at (pants.leftleg) [left=.00\cobwidth] {$f'$};
        \node (a) at (tube.bot) [left=.05\cobwidth] {};
        \draw (f.center) to (a.center) to [in=90,out=-90] (copants.belt);
        \draw (g.center) to (pants.belt);
        \draw (f.center) to [out=90,in=-165] (g.center);
        \draw (f.center) to [out=45,in=-90] ++(0.19,0.15) to ++(0,0.1) to [out=90,in=180] ++(0.09,0.1) to [out=0,in=90] ++(0.09,-0.1) to ++(0,-1.25) to [out=-90,in=180] ++(0.3,-0.2) to [out=0,in=-90] ++(0.3,0.2) to [out=90,in=-90] (copants.rightleg);
        \draw (g.center) to [out=-15,in=90] (pants.rightleg);
      \end{scope}
    \end{tikzpicture}
    \ \simopt \
    \begin{tikzpicture}[baseline={([yshift=-.5ex]current bounding box.center)}]
      \node[Pants, top] (pants) {};
      \node[Cyl, anchor=top] (tube) at (pants.leftleg){};
      \node[Copants, bot, height scale = 1.5, anchor=leftleg] (copants) at (tube.bot) {};
      \node[Top3D] at (copants.rightleg) {};
      \node[Bot3D] at (pants.rightleg) {};
      \node[Bot3D] at (tube.center) {};
      \begin{scope}[internal string scope]
        \node[sq tiny label] (g) at (pants.center) {$g'$};
        \node[sq tiny label] (f) at (copants.center) [left=.05\cobwidth] {$f'$};
        \draw (f.center) to [in=90,out=-90] (copants.belt);
        \draw (g.center) to (pants.belt);
        \draw (f.center) to [out=165,in=-90] (copants.leftleg) to (tube.top) to [out=90,in=-165] (g.center);
        \draw (f.center) to [out=45,in=-90] ++(0.19,0.15) to ++(0,0.1) to [out=90,in=180] ++(0.09,0.1) to [out=0,in=90] ++(0.09,-0.1) to ++(0,-0.1) to [out=-90,in=180] ++(0.09,-0.1) to [out=0,in=-90] ++(0.09,0.1) to (copants.rightleg);
        \draw (g.center) to [out=-15,in=90] (pants.rightleg);
      \end{scope}
    \end{tikzpicture}
    \ = \
    \begin{tikzpicture}[baseline={([yshift=-.5ex]current bounding box.center)}]
      \node[Pants, top] (pants) {};
      \node[Cyl, anchor=top] (tube) at (pants.leftleg){};
      \node[Copants, bot, anchor=leftleg] (copants) at (tube.bot) {};
      \node[Top3D] at (copants.rightleg) {};
      \node[Bot3D] at (pants.rightleg) {};
      \node[Bot3D] at (tube.center) {};
      \begin{scope}[internal string scope]
        \node[sq tiny label] (g) at (pants.center) {$g'$};
        \node[sq tiny label] (f) at (copants.center) {$f'$};
        \draw (f.center) to (copants.belt);
        \draw (g.center) to (pants.belt);
        \draw (f.center) to [out=165,in=-90] (copants.leftleg) to (tube.top) to [out=90,in=-165] (g.center);
        \draw (f.center) to [out=15,in=-90] (copants.rightleg);
        \draw (g.center) to [out=-15,in=90] (pants.rightleg);
      \end{scope}
    \end{tikzpicture}
  \end{equation*}

  So we have established that comb equivalence implies optic equivalence.
  This is sufficient to show that the functor of proposition \ref{prop:optic_comb} is also faithful.
\end{proof}
\begin{remark}
  The previous result could also be established by Yoneda reduction (see e.g.\ \cite[Sec.\ 4.2]{roman_comb}) as follows:
  \begin{align*}
    \int^E \cat{C}(A,E\otimes B)\times\cat{C}(E\otimes B',A') \cong \int^E \cat{C}(A,E\otimes B)\times\cat{C}(E,{B'}^*\otimes A') & \cong \cat{C}(A,{B'}^*\otimes A'\otimes B) \\
    & \cong \cat{C}(A\otimes B',A'\otimes B)
  \end{align*}
  Note that $(f,g)_E\simcomb (f',g')_{E'}$ implies $(f,g)_E\simsigma(f',g')_{E'}$ which ensures they are the same element of the set $\cat{C}(A\otimes B',A'\otimes B)$.
\end{remark}

\begin{prop}
 Given a Cartesian category $\cat{C}$ where each type is inhabited, there is a symmetric monoidal isomorphism of categories $\optic(\cat{C}) \cong \comb(\cat{C})$.
\end{prop}
\begin{proof}
  Suppose $(f,g)_E\simcomb (f',g')_{E'}$.
  We know that these combs are equal on the braid:
  \begin{equation*}
    \tikzfigscale{0.9}{figs/comb_braid2}
  \end{equation*}
  By the universal property of the product, this map is completely determined by its projections into $A'$ and $B$.
  The former gives:
  \begin{equation}\label{eq:cart1}
    \tikzfigscale{0.9}{figs/cart1} \in \cat{C}(A\times B',A')
  \end{equation}
  while the latter gives
  \begin{equation*}
    \tikzfigscale{0.9}{figs/cart2}
  \end{equation*}
Pick a map $\phi:I\to A$, then
  \begin{equation}\label{eq:cart3}
    \tikzfigscale{0.9}{figs/cart3}
  \end{equation}

  Thus:
  \begin{equation*}
    \begin{tikzpicture}[baseline={([yshift=-.5ex]current bounding box.center)}]
      \node[Pants, top] (pants) {};
      \node[Cyl, anchor=top] (tube) at (pants.leftleg){};
      \node[Copants, bot, anchor=leftleg] (copants) at (tube.bot) {};
      \node[Top3D] at (copants.rightleg) {};
      \node[Bot3D] at (pants.rightleg) {};
      \node[Bot3D] at (tube.center) {};
      \begin{scope}[internal string scope]
        \node[sq tiny label] (g) at (pants.center) {$g$};
        \node[sq tiny label] (f) at (copants.center) {$f$};
        \draw (f.center) to (copants.belt);
        \draw (g.center) to (pants.belt);
        \draw (f.center) to [out=165,in=-90] (copants.leftleg) to (tube.top) to [out=90,in=-165] (g.center);
        \draw (f.center) to [out=15,in=-90] (copants.rightleg);
        \draw (g.center) to [out=-15,in=90] (pants.rightleg);
      \end{scope}
    \end{tikzpicture}
    \ = \
    \begin{tikzpicture}[baseline={([yshift=-.5ex]current bounding box.center)}]
      \node[Pants, top, scale=1.5] (pants) {};
      \node[Cyl, anchor=top, scale=1.5] (tube) at (pants.leftleg){};
      \node[Copants, bot, anchor=leftleg, scale=1.5] (copants) at (tube.bot) {};
      \node[Top3D, scale=1.5] at (copants.rightleg) {};
      \node[Bot3D, scale=1.5] at (pants.rightleg) {};
      \node[Bot3D, scale=1.5] at (tube.center) {};
      \begin{scope}[internal string scope]
        \node[sq tiny label] (g) at (pants.center) {$g$};
        \node[sq tiny label] (delta) at ($(copants.center)+(0,-.35)$) {$\Delta$};
        \node[sq tiny label] (f0) at ($(delta.center)+(-.3,.25)$) {$f$};
        \node[sq tiny label] (f1) at ($(delta.center)+(.3,.25)$) {$f$};
        \node[sq tiny label] (pi0) at ($(f0.center)+(0,.35)$) {$!$};
        \node[sq tiny label] (pi1) at ($(f1.center)+(0,.35)$) {$!$};
        \draw (g.center) to (pants.belt);
        \draw (g.center) to [out=-15,in=90] (pants.rightleg);
        \draw (copants.belt) to (delta.center);
        \draw[bend right] (delta.center) to (f1.center);
        \draw[bend left] (delta.center) to (f0.center);
        \draw (f0.center) to (pi0.center);
        \draw (f1.center) to (pi1.center);
        \draw (f0.center) to [out=165,in=-90] (tube.bot) to (tube.top) to [out=90,in=-165] (g.center);
        \draw (f1.center) to [out=15,in=-90] (copants.rightleg);
      \end{scope}
    \end{tikzpicture}
    \ \sim \
    \begin{tikzpicture}[baseline={([yshift=-.5ex]current bounding box.center)}]
      \node[Pants, top, scale=1.5] (pants) {};
      \node[Cyl, anchor=top, scale=1.5] (tube) at (pants.leftleg){};
      \node[Copants, bot, anchor=leftleg, scale=1.5] (copants) at (tube.bot) {};
      \node[Top3D, scale=1.5] at (copants.rightleg) {};
      \node[Bot3D, scale=1.5] at (pants.rightleg) {};
      \node[Bot3D, scale=1.5] at (tube.center) {};
      \begin{scope}[internal string scope]
        \node[sq tiny label] (g) at ($(pants.center)+(-.1,.15)$) {$g$};
        \node[sq tiny label] (delta) at ($(copants.center)+(0,-.35)$) {$\Delta$};
        \node[sq tiny label] (f0) at ($(g.center)+(-.4,-.4)$) {$f$};
        \node[sq tiny label] (f1) at ($(delta.center)+(.3,.25)$) {$f$};
        \node[sq tiny label] (pi0) at ($(g.center)+(.38,.1)$) {$!$};
        \node[sq tiny label] (pi1) at ($(f1.center)+(0,.35)$) {$!$};
        \draw (g.center) to ($(pants.belt)+(-.1,0)$);
        \draw (g.center) to [out=-15,in=90] (pants.rightleg);
        \draw (copants.belt) to (delta.center);
        \draw[bend right] (delta.center) to (f1.center);
        \draw (delta.center) to [out=165,in=-90] ($(tube.bot)+(.1,0)$) to ($(tube.top)+(.1,0)$) to [out=90,in=-90] (f0.center);
        \draw[bend right] (f0.center) to (pi0.center);
        \draw (f1.center) to (pi1.center);
        \draw[bend left]  (f0.center) to  (g.center);
        \draw (f1.center) to [out=15,in=-90] (copants.rightleg);
      \end{scope}
    \end{tikzpicture}
    \ = \
    \begin{tikzpicture}[baseline={([yshift=-.5ex]current bounding box.center)}]
      \node[Pants, top, scale=1.5] (pants) {};
      \node[Cyl, anchor=top, scale=1.5] (tube) at (pants.leftleg){};
      \node[Copants, bot, anchor=leftleg, scale=1.5] (copants) at (tube.bot) {};
      \node[Top3D, scale=1.5] at (copants.rightleg) {};
      \node[Bot3D, scale=1.5] at (pants.rightleg) {};
      \node[Bot3D, scale=1.5] at (tube.center) {};
      \begin{scope}[internal string scope]
        \node[sq tiny label] (g) at ($(pants.center)+(-.125,.15)$) {$g'$};
        \node[sq tiny label] (delta) at ($(copants.center)+(0,-.35)$) {$\Delta$};
        \node[sq tiny label] (f0) at ($(g.center)+(-.38,-.4)$) {$f'$};
        \node[sq tiny label] (f1) at ($(delta.center)+(.3,.25)$) {$f'$};
        \node[sq tiny label] (pi0) at ($(g.center)+(.38,.1)$) {$!$};
        \node[sq tiny label] (pi1) at ($(f1.center)+(0,.35)$) {$!$};
        \draw (g.center) to ($(pants.belt)+(-.125,0)$);
        \draw (g.center) to [out=-15,in=90] (pants.rightleg);
        \draw (copants.belt) to (delta.center);
        \draw[bend right] (delta.center) to (f1.center);
        \draw (delta.center) to [out=165,in=-90] ($(tube.bot)+(.1,0)$) to ($(tube.top)+(.1,0)$) to [out=90,in=-90] (f0.center);
        \draw[bend right] (f0.center) to (pi0.center);
        \draw (f1.center) to (pi1.center);
        \draw[bend left]  (f0.center) to  (g.center);
        \draw (f1.center) to [out=15,in=-90] (copants.rightleg);
      \end{scope}
    \end{tikzpicture}
    \ \sim \
    \begin{tikzpicture}[baseline={([yshift=-.5ex]current bounding box.center)}]
      \node[Pants, top, scale=1.5] (pants) {};
      \node[Cyl, anchor=top, scale=1.5] (tube) at (pants.leftleg){};
      \node[Copants, bot, anchor=leftleg, scale=1.5] (copants) at (tube.bot) {};
      \node[Top3D, scale=1.5] at (copants.rightleg) {};
      \node[Bot3D, scale=1.5] at (pants.rightleg) {};
      \node[Bot3D, scale=1.5] at (tube.center) {};
      \begin{scope}[internal string scope]
        \node[sq tiny label] (g) at (pants.center) {$g'$};
        \node[sq tiny label] (delta) at ($(copants.center)+(0,-.35)$) {$\Delta$};
        \node[sq tiny label] (f0) at ($(delta.center)+(-.3,.25)$) {$f'$};
        \node[sq tiny label] (f1) at ($(delta.center)+(.3,.25)$) {$f'$};
        \node[sq tiny label] (pi0) at ($(f0.center)+(0,.35)$) {$!$};
        \node[sq tiny label] (pi1) at ($(f1.center)+(0,.35)$) {$!$};
        \draw (g.center) to (pants.belt);
        \draw (g.center) to [out=-15,in=90] (pants.rightleg);
        \draw (copants.belt) to (delta.center);
        \draw[bend right] (delta.center) to (f1.center);
        \draw[bend left] (delta.center) to (f0.center);
        \draw (f0.center) to (pi0.center);
        \draw (f1.center) to (pi1.center);
        \draw (f0.center) to [out=165,in=-90] (tube.bot) to (tube.top) to [out=90,in=-165] (g.center);
        \draw (f1.center) to [out=15,in=-90] (copants.rightleg);
      \end{scope}
    \end{tikzpicture}
    =
    \begin{tikzpicture}[baseline={([yshift=-.5ex]current bounding box.center)}]
      \node[Pants, top] (pants) {};
      \node[Cyl, anchor=top] (tube) at (pants.leftleg){};
      \node[Copants, bot, anchor=leftleg] (copants) at (tube.bot) {};
      \node[Top3D] at (copants.rightleg) {};
      \node[Bot3D] at (pants.rightleg) {};
      \node[Bot3D] at (tube.center) {};
      \begin{scope}[internal string scope]
        \node[sq tiny label] (g) at (pants.center) {$g'$};
        \node[sq tiny label] (f) at (copants.center) {$f'$};
        \draw (f.center) to (copants.belt);
        \draw (g.center) to (pants.belt);
        \draw (f.center) to [out=165,in=-90] (copants.leftleg) to (tube.top) to [out=90,in=-165] (g.center);
        \draw (f.center) to [out=15,in=-90] (copants.rightleg);
        \draw (g.center) to [out=-15,in=90] (pants.rightleg);
      \end{scope}
    \end{tikzpicture}
  \end{equation*}
\end{proof}
\begin{remark}
  The final part of the proof can also be derived by Yoneda reduction (see e.g.\ \cite[Sec.\ 3.1]{clarke_profunctor}):
    \begin{align*}
      \int^E \cat{C}(A,E\times B)\times\cat{C}(E\times B',A') & \cong \int^E\cat{C}(A,E)\times\cat{C}(A,B)\times\cat{C}(E\times B',A')
       \cong \cat{C}(A,B)\times \cat{C}(A\times B',A')
    \end{align*}
  and then noting that the projections \eqref{eq:cart1} and \eqref{eq:cart3} precisely determine an element of $\cat{C}(A,B)\times \cat{C}(A\times B',A')$.
\end{remark}

\begin{prop}
There is a symmetric monoidal isomorphism $\optic(\sf{Unitary}) \cong \comb(\sf{Unitary})$, where $\sf{Unitary}$ is the category of unitary maps between (not necessarily finite dimensional) Hilbert spaces.
\end{prop}
\begin{proof}
  $f:A\morph{}E\otimes B$ is a unitary and thus $A\cong E\otimes B$ are isomorphic as Hilbert spaces.
  Similarly from $f'$ we see $A\cong E'\otimes B$ and from $g$ and $g'$, $A'\cong E\otimes B' \cong E'\otimes B'$.
  This means there must exist a unitary $U:E\otimes B\morph{} E'\otimes B$ such that $f' = Uf$ and a unitary $V:E'\otimes B'\morph{} E\otimes B'$ such that $g' = gV$.

  Using the fact that $(f,g)_E\simcomb (f',g')_{E'}$ and that $f$ and $g$ have two-sided inverses, we see that for all $\lambda$:
  \begin{equation}\label{eq:unitaries}
    \tikzfigscale{0.9}{figs/unitaries1}
  \end{equation}
  Taking $\lambda = \sigma$ we arrive at the following equality:
  \begin{equation*}
    \tikzfigscale{0.9}{figs/unitaries2}
  \end{equation*}
  There exists a faithful embedding of $\mathsf{Unitaries}$ into $\mathsf{Hilb}$ where we can pick any state $\ket{\psi}$ and effect $\bra{e}$ with $\bra{e}\ket{\psi}= 1$ to see that:
  \begin{equation*}
    \tikzfigscale{0.9}{figs/unitaries3}
  \end{equation*}
  As a result $U$ can be seen to $\otimes$-separate as $U=U'\otimes 1$ where $U':=(1\otimes e)V^{-1}(1\otimes \psi)$ must be a unitary else $U$ could not be unitary and we would have a contradiction.
  Analagously one can show that $V$ $\otimes$-separates as $V'\otimes 1$.
  Inserting these factorisations into the right hand side of \eqref{eq:unitaries} one can see that $V'U' = 1$.

  Therefore:
  \begin{equation*}
    \begin{tikzpicture}[baseline={([yshift=-.5ex]current bounding box.center)}]
      \node[Pants, top] (pants) {};
      \node[Cyl, anchor=top] (tube) at (pants.leftleg){};
      \node[Copants, bot, lowercob, anchor=leftleg] (copants) at (tube.bot) {};
      \node[Top3D] at (copants.rightleg) {};
      \node[Bot3D] at (pants.rightleg) {};
      \node[Bot3D] at (tube.center) {};
      \begin{scope}[internal string scope]
        \node[sq tiny label] (g) at (pants.center) {$g'$};
        \node[sq tiny label] (f) at (copants.center) {$f'$};
        \draw (f.center) to (copants.belt);
        \draw (g.center) to (pants.belt);
        \draw (f.center) to [out=165,in=-90] (copants.leftleg) to (tube.top) to [out=90,in=-165] (g.center);
        \draw (f.center) to [out=15,in=-90] (copants.rightleg);
        \draw (g.center) to [out=-15,in=90] (pants.rightleg);
      \end{scope}
    \end{tikzpicture}
    \ = \
    \begin{tikzpicture}[baseline={([yshift=-.5ex]current bounding box.center)}]
      \node[Pants, top] (pants) {};
      \node[Cyl, anchor=top] (tube) at (pants.leftleg){};
      \node[Copants, bot, lowercob, anchor=leftleg] (copants) at (tube.bot) {};
      \node[Top3D] at (copants.rightleg) {};
      \node[Bot3D] at (pants.rightleg) {};
      \node[Bot3D] at (tube.center) {};
      \begin{scope}[internal string scope]
        \node[sq tiny label] (g) at (pants.center) {$g$};
        \node[sq tiny label] (f) at (copants.center) {$f$};
        \node[sq tiny label] (u) at (copants.leftleg) {$U'$};
        \node[sq tiny label] (v) at (pants.leftleg) {$V'$};
        \draw (f.center) to (copants.belt);
        \draw (g.center) to (pants.belt);
        \draw (f.center) to [out=165,in=-90] (copants.leftleg) to (tube.top) to [out=90,in=-165] (g.center);
        \draw (f.center) to [out=15,in=-90] (copants.rightleg);
        \draw (g.center) to [out=-15,in=90] (pants.rightleg);
      \end{scope}
    \end{tikzpicture}
    \ \simopt \
    \begin{tikzpicture}[baseline={([yshift=-.5ex]current bounding box.center)}]
      \node[Pants, top] (pants) {};
      \node[Cyl, anchor=top] (tube) at (pants.leftleg){};
      \node[Copants, bot, lowercob, anchor=leftleg] (copants) at (tube.bot) {};
      \node[Top3D] at (copants.rightleg) {};
      \node[Bot3D] at (pants.rightleg) {};
      \node[Bot3D,yshift=-0.1cm] at (tube.center) {};
      \begin{scope}[internal string scope]
        \node[sq tiny label] (g) at (pants.center) {$g$};
        \node[sq tiny label] (f) at (copants.center) {$f$};
        \node[sq tiny label,yshift=0.5cm] (u) at (copants.leftleg) {$U'$};
        \node[sq tiny label,yshift=0.1cm] (v) at (pants.leftleg) {$V'$};
        \draw (f.center) to (copants.belt);
        \draw (g.center) to (pants.belt);
        \draw (f.center) to [out=165,in=-90] (copants.leftleg) to (tube.top) to [out=90,in=-165] (g.center);
        \draw (f.center) to [out=15,in=-90] (copants.rightleg);
        \draw (g.center) to [out=-15,in=90] (pants.rightleg);
      \end{scope}
    \end{tikzpicture}
    \ = \
    \begin{tikzpicture}[baseline={([yshift=-.5ex]current bounding box.center)}]
      \node[Pants, top] (pants) {};
      \node[Cyl, anchor=top] (tube) at (pants.leftleg){};
      \node[Copants, bot, lowercob, anchor=leftleg] (copants) at (tube.bot) {};
      \node[Top3D] at (copants.rightleg) {};
      \node[Bot3D] at (pants.rightleg) {};
      \node[Bot3D] at (tube.center) {};
      \begin{scope}[internal string scope]
        \node[sq tiny label] (g) at (pants.center) {$g$};
        \node[sq tiny label] (f) at (copants.center) {$f$};
        \draw (f.center) to (copants.belt);
        \draw (g.center) to (pants.belt);
        \draw (f.center) to [out=165,in=-90] (copants.leftleg) to (tube.top) to [out=90,in=-165] (g.center);
        \draw (f.center) to [out=15,in=-90] (copants.rightleg);
        \draw (g.center) to [out=-15,in=90] (pants.rightleg);
      \end{scope}
    \end{tikzpicture}
  \end{equation*}
\end{proof}

\section{The CPM construction as optics}\label{sec:CPM}
In this section we show that over a $\dag$-compact closed category, the $\cpm$ construction embeds within optics.

\begin{defn}[CPM construction \cite{selinger_cpm}]
  Given a $\dag$-compact closed category $\cat{C}$, the category $\cpm(\cat{C})$ of completely positive maps has the same objects as $\cat{C}$.
  A morphism $f:A\morph{}B$ in $\cpm(\cat{C})$  is a morphism of type $A^*\otimes A\morph{}B^*\otimes B$ in $\cat{C}$ of the form
    \begin{equation}\label{eq:cpm}
      \tikzfigscale{0.9}{figs/cpm}
    \end{equation}
  where $(-)^*:\cat{C}\to\cat{C}$ is the conjugation functor.
  Composition and identities are inherited from $\cat{C}$.
\end{defn}

\begin{example}
The $\dag$-symmetric monoidal category $\cpm(\fhilb)$ is equivalent to the category of density operators between finite dimensional Hilbert spaces.
\end{example}

The $\optic$ and $\comb$ constructions provide another route to defining the category of completely positive maps.
We write $\doptic(\cat{C})$ and $\dcomb(\cat{C})$ for the subcategories of $\optic(\cat{C})$ and $\comb(\cat{C})$ respectively, generated by representatives of the form $(f,f^\dag)_E$.
The following proposition follows:

\begin{prop}\label{prop:cpmequiv}
  When $\cat{C}$ is $\dag$-compact closed, there is a symmetric monoidal isomorphism of categories $\doptic(\cat{C})\cong\dcomb(\cat{C})\cong\cpm(\cat{C})$.
\end{prop}

\begin{proof}(Sketch).
  The isomorphism $\doptic(\cat{C})\cong\dcomb(\cat{C})$ follows by proposition \ref{prop:compactclosed}.
  The isomorphism $\dcomb(\cat{C})\cong\cpm(\cat{C})$ is given by sending $(A,A)\mapsto A$ on objects and $(f,f^\dag)_E$ to \eqref{eq:cpm}.
  Fullness is obvious and one can see it is faithful by inserting the braid into the comb equivalence relation \eqref{eq:combequiv}.
\end{proof}

There have been attempts to generalise the $\cpm$ construction to infinite dimensional quantum systems where one does not have compact closure.
For instance, in \cite{coecke_cpinf} the $\mathsf{CP}^\infty$ construction is developed which turns any monoidal $\dag$-category into a category of completely positive maps.
The category $\mathsf{CP}^\infty(\cat{C})$ is very similar to the category $\comb(\cat{C})$, the only difference being that $\mathsf{CP}^\infty$ only quantifies over the positive maps in the equivalence relation \eqref{eq:combequiv} (as opposed to all maps $\lambda$).
In the case that $\cat{C}$ is $\dag$-compact closed it is known that $\cpm(\cat{C})\cong\mathsf{CP}^\infty(\cat{C})$ and thus the $\mathsf{CP}^\infty$ construction produces an isomorphic category to $\doptic$ and $\dcomb$.
Dropping compact closure, but keeping the $\dag$-symmetric monoidal structure, $\doptic$ and $\dcomb$ yield two potential candiates for generalised categories of completely positive maps.

\section{$n$-Combs} \label{sec:ncombs}
In this section we consider generalisations of the $\optic$ and $\comb$ constructions to encompass $n$-combs.
There are several categorical structures that could provide an adequate semantics for dealing with the many inputs and outputs that a generalised $n$-comb could have.
Here we will use polycategories to handle $n$-combs.

A candidate definition of such an $n$-comb was suggested in \cite{roman_comb} as a generalisation of the $\optic$ construction.
We generalise this even further, obtaining a polycategory.
Our definition of the combs themselves is similar, but crucially our notion of composition is very different and coincides more closely with that of \cite{coecke_resources}.

\begin{defn}
  Given a symmetric monoidal category $\cat C$, the polycategory of $n$-combs $\OPTIC(\cat{C})$ has:
  \begin{description}
  \item[Objects:] Pairs of objects in $\cat C$.
  \item[Morphisms:]  The polymorphisms of type $[(A_1,A_1'),\ldots , (A_n,A_n')] \to [(B_1,B_1'),\ldots,(B_m,B_m')]$ are elements of the set (where the zero-fold tensor in $\cat C$ is the tensor unit):
  \begin{equation*}
   \int^{X_0,\ldots, X_{n+1}}
  \cat{C}\left(\bigotimes_{i=1}^n B_i, X_0\right) \times
   \prod_{i=1}^n \cat{C}(X_{i-1},X_{i}\otimes A_i)\times \cat{C}(X_{i}\otimes A_i', X_{i+1}) \times
  \cat{C}\left(X_{n+1},\bigotimes_{i=1}^n B_i'\right)
  \end{equation*}
  For example consider the string diagram for a polymorphism of this type (drawn from left to right to conserve space):

  \hspace*{-1cm}
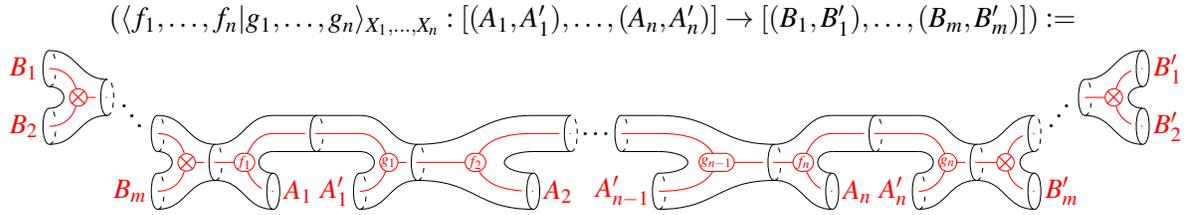
  \begin{tabular}{c}
  $\left(\langle f_1, \ldots, f_n | g_1, \ldots, g_n \rangle_{X_1,\ldots, X_n}: [(A_1,A_1'),\ldots, (A_n,A_n')] \to [(B_1,B_1'),\ldots, (B_m,B_m')]  \right) :=$ \\
  $
  \scalebox{.96}{
  \rotatebox[origin=c]{-90}{
    \begin{tikzpicture}[baseline={([yshift=-.5ex]current bounding box.center)}]
      \node[Pants] (pants)  {};
      \node[Cyl, anchor=top,yscale=1.5] (tube) at (pants.leftleg) {};                                       %tube
      \node[Copants, bot, anchor=leftleg] (copants) at (tube.bot) {};                    %copants
      \node[Top3D] at (copants.rightleg) {};
      \node[Bot3D] at (pants.rightleg) {};
      \node[Bot3D] at (tube.center) {};
      \node[Copants,anchor=belt,yscale=2] (copants2) at (pants.belt) {};      %copants2
      \node[Bot3D] at (pants.belt) {};
      \node[Cyl,anchor=bot,height scale=0.65] (tube2) at (copants2.leftleg) {};     %tube2
      \node[Bot3D] at (tube2.top) {};
      \node[Top3D] at (copants2.rightleg) {};
      \node (dots) at ($(tube2.top) + (0,0.45)$) {$\vdots$};
      \node[rotate=90]  at ($(tube2.top) + (-.3,0.35)$) {}; %This should be n or something
      \node[Cyl,anchor=bot,height scale=0.65] (tube3) at ($(dots.center)+(0,0.25)$) {}; %tube3
      \node[Bot3D] at (tube3.bot) {};
      \node[Pants,anchor=leftleg,yscale=2] (pants2) at (tube3.top) {};                 %pants2
      \node[Bot3D] at (pants2.rightleg) {};
      \node[Copants,bot,anchor=belt] (copants3) at (pants2.belt) {};                   %copants3
      \node[Cyl, anchor=bot,yscale=1.5] (tube4) at (copants3.leftleg) {};                               %tube4
      \node[Pants,anchor=leftleg] (pants3) at (tube4.top) {};                          %pants3
      \node[Bot3D] at (pants3.belt) {};
      \node[Top3D] at (copants3.rightleg) {};
      \node[Bot3D] at (tube4.center) {};
      \node[Bot3D] at (pants3.rightleg) {};
      \node[Copants,anchor=belt]  (copants4) at (pants3.belt) {};                        %copants4
      \node[Bot3D] at (copants4.leftleg) {};
      \node[Top3D] at (copants4.rightleg) {};
      \node (dots1) at ($(copants4.leftleg) + (-.25,0.45)$) {$\ddots$};
      \node[rotate=90]  at  ($(copants4.leftleg) + (-.5,0.2)$) {}; %This should be m or something
      \node[Copants,top,anchor=belt]  (copants5) at ($(dots1)  + (-.25,0.25)$) {};                        %copants5
      \node[Bot3D] at (copants5.belt) {};
      \node[Pants,anchor=belt]  (pants4) at (copants.belt) {};                        %pants4
      \node[Bot3D] at (pants4.leftleg) {};
      \node[Bot3D] at (pants4.rightleg) {};
      \node(dots2)  at ($(pants4.leftleg)+(-.25,-0.25)$)  {\reflectbox{$\ddots$}};
      \node[rotate=90]  at  ($(pants4.leftleg) + (-.5,-0.2)$) {}; %This should be m or something
      \node[Pants,bot,anchor=belt]  (pants5) at  ($(dots2)  + (-.25,-0.45)$){};                        %pants5
      \node[Bot3D] at (pants5.belt) {};
      \begin{scope}[internal string scope]
       %morphisms
        \node[sq tiny label,rotate=90] (f0) at (copants.center) {$f_1$};
        \node[sq tiny label,rotate=90] (f1) at (copants2.center) {$f_2$};
        \node[sq tiny label,rotate=90] (fn) at (copants3.center) {$f_n$};
        \node[sq tiny label,rotate=90] (g0) at (pants.center) {$g_1$};
        \node[sq tiny label,rotate=90] (gn1) at (pants2.center) {$g_{n-1}$};
        \node[sq tiny label,rotate=90] (gn) at (pants3.center) {$g_n$};
        %tensor nodes
        \node[tiny label, inner sep=0cm,minimum size=.24cm, fill opacity=0] (cotensor) at (copants4.center) {};
        \node at  (cotensor.center) {$\times$};
        \node[tiny label, inner sep=0cm,minimum size=.24cm, fill opacity=0] (cotensor1) at (copants5.center) {$\otimes$};
        \node at  (cotensor1.center) {$\times$};
        \node[tiny label, inner sep=0cm,minimum size=.24cm, fill opacity=0] (tensor) at (pants4.center) {$\otimes$};
        \node at  (tensor.center) {$\times$};
        \node[tiny label, inner sep=0cm,minimum size=.24cm, fill opacity=0] (tensor1) at (pants5.center) {$\otimes$};
        \node at  (tensor1.center) {$\times$};
        %wires
        \draw (f0.center) to (copants.belt);
        \draw (g0.center) to (pants.belt);
        \draw (f0.center) to [out=165,in=-90] (copants.leftleg) to (tube.top) to [out=90,in=-165] (g0.center);
        \draw (f0.center) to [out=15,in=-90] (copants.rightleg);
        \draw (g0.center) to [out=-15,in=90] (pants.rightleg);
        \draw (f1.center) to (copants2.belt);
        \draw (gn1.center) to (pants2.belt);
        \draw (f1.center) to [out=165,in=-90] (copants2.leftleg) to (tube2.top);
        \draw (tube3.bot) to (pants2.leftleg) to [out=90,in=-165] (gn1.center);
        \draw (f1.center) to [out=15,in=-90] (copants2.rightleg);
        \draw (gn1.center) to [out=-15,in=90] (pants2.rightleg);
        \draw (fn.center) to (copants3.belt);
        \draw (fn.center) to [out=165,in=-90] (copants3.leftleg) to (tube4.top) to [out=90,in=-165] (gn.center);
        \draw (fn.center) to [out=15,in=-90] (copants3.rightleg);
        \draw (gn.center) to [out=-15,in=90] (pants3.rightleg);
        \draw (gn.center) to (pants3.belt);
        \draw (pants3.belt) to (cotensor);
        \draw (cotensor) to  [out=165,in=-90] (copants4.leftleg);
        \draw (cotensor) to  [out=15,in=-90] (copants4.rightleg);
        \draw (copants5.belt) to (cotensor1);
        \draw (cotensor1) to  [out=165,in=-90] (copants5.leftleg);
        \draw (cotensor1) to  [out=15,in=-90] (copants5.rightleg);
        \draw (copants.belt) to (tensor);
        \draw (tensor) to  [out=-165,in=90] (pants4.leftleg);
        \draw (tensor) to  [out=-15,in=90] (pants4.rightleg);
        \draw (pants5.belt) to (tensor1);
        \draw (tensor1) to  [out=-165,in=90] (pants5.leftleg);
        \draw (tensor1) to  [out=-15,in=90] (pants5.rightleg);
        %Object types
        \node[rotate=90] at ($(pants5.leftleg)+(0,-.35)$) {$B_1$};
        \node[rotate=90] at ($(pants5.rightleg)+(0,-.35)$) {$B_2$};
        \node[rotate=90] at ($(pants4.rightleg)+(0,-.38)$) {$B_m$};
        \node[rotate=90] at ($(copants5.leftleg)+(0,.35)$) {$B_1'$};
        \node[rotate=90] at ($(copants5.rightleg)+(0,.35)$) {$B_2'$};
        \node[rotate=90] at ($(copants4.rightleg)+(0,.38)$) {$B_m'$};
        \node[rotate=90] at ($(copants.rightleg)+(0,.35)$) {$A_1$};
        \node[rotate=90] at ($(copants2.rightleg)+(0,.35)$) {$A_2$};
        \node[rotate=90] at ($(copants3.rightleg)+(0,.35)$) {$A_n$};
        \node[rotate=90] at ($(pants.rightleg)+(0,-.35)$) {$A_1'$};
        \node[rotate=90] at ($(pants2.rightleg)+(0,-.5)$) {$A_{n-1}'$};
        \node[rotate=90] at ($(pants3.rightleg)+(0,-.35)$) {$A_n'$};
      \end{scope}
    \end{tikzpicture}}}$
  \end{tabular}

  The identities are the same as in optics.
  Given a map as above and another map

  $$\langle h_0, \ldots, h_{\ell} | k_0, \ldots, k_n \rangle_{Y_1,\ldots, Y_{\ell}}: [(C_1, C_1'), \ldots, (C_\ell, C_\ell')] \to [(D_1 ,D_1'), \ldots, (D_p,D_p') ]$$
   where $(C_q,C_q') =(B_j,B_j')$ for some $0 \leq q \leq \ell$, $0 \leq j \leq m$.
  Then the composite
  $$\langle f_1, \ldots, f_n | g_1, \ldots, g_n \rangle_{X_1,\ldots, X_n} \circ_{(B_j,B_j')} \langle h_0, \ldots, h_{\ell} | k_0, \ldots, k_n \rangle_{Y_1,\ldots, Y_{\ell}}$$

  is given by plugging the first comb into the $(B_j,B_j')$ hole and the collapsing the bubble.  This can be verified to produce a diagram of the same shape via a lengthy, yet elementary application of the coend calculus, or equivalently a composition of the 2-cells \eqref{eq:tube2cells} and associators.

  \end{description}
\end{defn}

There is also a polycategory of $n$-combs that generalises the $\comb$ construction,
\begin{defn}
  Given a symmetric monoidal category $\cat{C}$, the polycategory of $n$-combs $\COMB(\cat{C})$ has the same objects as $\OPTIC(\cat{C})$.
  The polymorphisms are given by tuples of maps under a generalisation of the comb equivalence relation where two combs are equivalent if they are equal on all extended inputs:
  \begin{equation*}
    \tikzfigscale{0.9}{figs/ncombequiv}
  \end{equation*}
  Composition and identities are the same as in $\COMB(\cat{C})$.
\end{defn}

As in the case of 1-combs we can always quotient the intensional optics definition to get the extensional comb definition:
\begin{prop}
  \label{prop:polyfunctor}
  There is a full and bijective on objects polyfunctor $\OPTIC(\cat{C})\morph{}\COMB(\cat{C})$.
\end{prop}
\begin{proof}(Sketch)
  The proof is similar to proposition \ref{prop:optic_comb}: removing the cobordisms and evaluating the comb on the given $\lambda_1,\dots,\lambda_n$ gives a cowedge and thus factorises uniquely via the coend.
\end{proof}

\begin{lem}
When $\cat C$ is compact closed, $\OPTIC(\cat{C})$ and $\COMB(\cat{C})$ are $*$-polycategories.
\end{lem}
\begin{proof}(Sketch)
In $\OPTIC(\cat{C})$, the unit and counits which generate the *-polycategory structure are given by the following (rotated) internal string diagrams:

$$
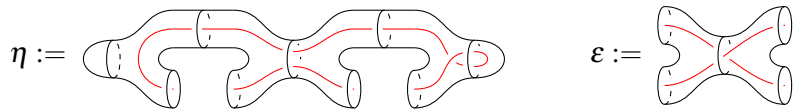
\eta:=
\rotatebox[origin=c]{-90}{
   \begin{tikzpicture}[baseline={([yshift=-.5ex]current bounding box.center)}]
      \node[Pants] (pants) {};
      \node[Cyl, anchor=top] (tube) at (pants.leftleg){};
      \node[Copants, bot, anchor=leftleg] (copants) at (tube.bot) {};
      \node[Top3D] at (copants.rightleg) {};
      \node[Bot3D] at (pants.rightleg) {};
      \node[Bot3D] at (tube.center) {};
      \node[Cap,bot] at (pants.belt) {};
      \node[Pants, anchor=belt] (pants1) at (copants.belt) {};
      \node[Cyl,anchor=top] (tube1) at (pants1.leftleg) {};
      \node[Copants, bot, anchor=leftleg] (copants1) at (tube1.bot) {};
      \node[Bot3D] at (pants1.rightleg) {};
      \node[Bot3D] at (tube1.center) {};
      \node[Top3D] at (copants1.rightleg) {};
      \node[Cup,bot] at (copants1.belt) {};
     \begin{scope}[internal string scope]
      \draw [bend left=90, looseness=3] ($(pants.belt)+(-.1,0)$) to ($(pants.belt)+(.1,0)$);
      \draw [in=90, out=-90, looseness=1.5] ($(pants.belt)+(.1,0)$) to (pants.leftleg);
      \draw [in=90, out=-90, looseness=1.5] ($(pants.belt)+(-.1,0)$) to (pants.rightleg);
      \draw (pants.leftleg) to (copants.leftleg);
      \draw [in=90, out=-90] (copants.leftleg) to ($(pants1.belt)+(-.1,0)$);
      \draw [in=90, out=-90] (copants.rightleg) to ($(pants1.belt)+(.1,0)$);
      \draw [in=90, out=-90] ($(pants1.belt)+(-.1,0)$) to (pants1.leftleg);
      \draw [in=90, out=-90]  ($(pants1.belt)+(.1,0)$) to (pants1.rightleg);
      \draw (pants1.leftleg) to (copants1.leftleg);
      \draw [bend right=90, looseness=2] (copants1.leftleg) to (copants1.rightleg);
     \end{scope}
   \end{tikzpicture}}
\hspace*{1cm}
\varepsilon:=
\rotatebox[origin=c]{-90}{
  \begin{tikzpicture}[baseline={([yshift=-.5ex]current bounding box.center)}]
    \node[Pants,xscale=1,bot] (pants) {};
    \node[Copants,xscale=1, top, anchor=belt] (copants) at (pants.belt) {};
    \node[Bot3D,xscale=1] at (pants.belt) {};
   \begin{scope}[internal string scope]
    \draw [in=90, out=-90, looseness=0.75] (copants.rightleg) to (pants.leftleg);
    \draw [in=90, out=-90, looseness=0.75] (copants.leftleg) to (pants.rightleg);
   \end{scope}
  \end{tikzpicture}
}
$$

 The $*$-polycategory structure of $\COMB(\cat{C})$ is transported along the polyfunctor in Proposition \ref{prop:polyfunctor}.
\end{proof}

\begin{prop}
  When $\cat{C}$ is compact closed there is an isomorphism of polycategories $\OPTIC(\cat{C}) \cong \COMB(\cat{C})$.
\end{prop}
\begin{proof}(Sketch)
The isomorphism is shown in a similar way to the proof of Proposition \ref{prop:compactclosed}, by pulling all of the circuits into the same bubble.
\end{proof}

\section{Conclusion and future work}
In this article we have considered some categorical approaches to modelling combs with a particular focus on the operational motivations typically pursued in the quantum literature.
There are several lines of future work we are actively investigating:

\begin{itemize}
  \item It would be clarifying to pin down precisely when $\optic(\cat{C})\cong\comb(\cat{C})$, or at least know whether this holds in cases beyond the few investigated here.
  Particularly for quantum theory we would like to know what happens in the case of $\mathsf{Isometry}$ and $\mathsf{CPTP}$.
  The cases of $*$-autonomous categories and monoidally closed categories would also be interesting so we could better understand any connections with the $\mathsf{Caus}$-construction \cite{kissinger_caus}.
  \item It may be possible to use profunctors to capture the causal structure of maps.
  Informally, one can replace causal graphs with profunctor tubes whose topology acts to restrict the families of maps that are compatible with the causal structure, for instance by enforcing one-way signalling constraints.
  \item We have reason to believe that the category of Tambara modules \cite{tambara} (equivalently the presheaf category of optics \cite{pastro_street}) is a good setting for modelling quantum supermaps more generally, possibly allowing for the modelling of maps like the quantum \texttt{switch} alongside combs.
  There are several operational principles one might ask of a quantum supermap to ensure that it is compatible with the monoidal structure of the category of first-order maps.  These principles seem to translate pleasingly into the structure of Tambara module homomorphisms.
  \item Given a $\cat{V}$-enriched category $\cat{C}$, it is not clear to us whether the category $\comb(\cat{C})$ inherits the enrichment.
  This is relevant to quantum theory because taking probabilistic mixures of quantum processes can be modelled by enrichment in the category $\mathsf{CMon}$ of commutative monoids \cite{heunen_vicary,gogioso_cpt}.
  On the other hand, it is immediate that $\optic(\cat{C})$ inherits the enrichment of $\cat{C}$ and thus might be a better setting for modelling quantum combs.
  \item In section \ref{sec:ncombs} we provided a polycategorical semantics for $n$-combs.
  We are also persuing the possibility of a double categorical framework which might be cleaner and possibly more expressive.
  Indeed, as we were writing this article we became aware of \cite{boisseau_cornering} where such a framework is developed.
\end{itemize}

\subsection*{Acknowledgements}
Many thanks to the authors of \cite{vicary_bordism} for use of their TikZ package for producing internal string diagrams and to Guillaume Boisseau and Matt Wilson for helpful discussions and comments. JH is supported by University College London and the EPSRC [grant number EP/L015242/1].

\bibliographystyle{eptcs}
\bibliography{bibliography}

%\appendix

\end{document}